\documentclass[11pt,letterpaper]{article}

\usepackage[margin=1in]{geometry}
\usepackage[T1]{fontenc}
\usepackage[utf8]{inputenc}
\usepackage{lmodern}
\usepackage{microtype}
\usepackage{amsmath,amssymb,amsthm,mathtools}
\usepackage{thm-restate}
\usepackage{graphicx}
\usepackage{xcolor}
\usepackage{booktabs}
\usepackage{enumitem}
\usepackage{tikz}
\usetikzlibrary{calc}
\usepackage[colorlinks=true,linkcolor=blue!55!black,citecolor=blue!55!black,urlcolor=blue!55!black]{hyperref}
\usepackage[section]{placeins}

\setlist[itemize]{leftmargin=2em,itemsep=0.2em,topsep=0.3em}
\setlist[enumerate]{leftmargin=2em,itemsep=0.2em,topsep=0.3em}

\newtheorem{theorem}{Theorem}[section]
\newtheorem{proposition}[theorem]{Proposition}

\newtheorem{corollary}[theorem]{Corollary}
\newtheorem{definition}[theorem]{Definition}

\newcommand{\hide}[1]{}
\newcommand{\Rcal}{\mathcal R}
\newcommand{\Fcal}{\mathcal F}
\newcommand{\Pcal}{\mathcal P}

\newcommand{\container}{\text{container}}
\newcommand{\filterbed}{\text{filter-bed}}
\newcommand{\dotcup}{\mathbin{\dot\cup}}
\newcommand{\1}{\mathbf 1}
\DeclareMathOperator{\LP}{LP}
\DeclareMathOperator{\OPT}{OPT}

\tikzset{
  bound/.style={draw=black,line width=1.05pt},
  port/.style={draw=black!70,fill=black!25,fill opacity=0.22,line width=0.5pt},
  portlabel/.style={above,font=\scriptsize,inner sep=1pt},
  rectlabel/.style={font=\scriptsize,inner sep=0.5pt,text opacity=1},
  rectA/.style={draw=blue!70!black,fill=blue!45,fill opacity=0.32,line width=0.65pt},
  rectB/.style={draw=orange!80!black,fill=orange!55,fill opacity=0.32,line width=0.65pt},
  rectC/.style={draw=green!50!black,fill=green!45,fill opacity=0.32,line width=0.65pt},
  rectD/.style={draw=purple!70!black,fill=purple!45,fill opacity=0.32,line width=0.65pt},
  rectE/.style={draw=red!70!black,fill=red!45,fill opacity=0.32,line width=0.65pt},
  rectP/.style={draw=teal!70!black,fill=teal!35,fill opacity=0.32,line width=0.65pt},
  rectQ/.style={draw=teal!70!black,fill=teal!35,fill opacity=0.32,line width=0.65pt},
  rectU/.style={draw=gray!80!black,fill=gray!45,fill opacity=0.32,line width=0.65pt},
  hcontainer/.style={draw=blue!75!black,fill=blue!10,fill opacity=0.16,line width=0.8pt},
  vcontainer/.style={draw=orange!85!black,fill=orange!18,fill opacity=0.16,line width=0.8pt},
  hport/.style={draw=blue!70!black,fill=blue!50,fill opacity=0.32,line width=0.45pt},
  vport/.style={draw=orange!85!black,fill=orange!55,fill opacity=0.32,line width=0.45pt},
  gridlabel/.style={font=\Large,inner sep=1.5pt,fill=white,fill opacity=0.85,text opacity=1}
}

\newcommand{\DrawPackage}[1]{%
\begin{tikzpicture}[x=2.5mm,y=3.2mm]
  \pgfmathtruncatemacro{\m}{#1}
  \pgfmathtruncatemacro{\kports}{2*\m}
  \pgfmathtruncatemacro{\last}{\m-1}
  \pgfmathtruncatemacro{\xmax}{8*\m+3}
  \pgfmathtruncatemacro{\ymax}{5*\m}
  \pgfmathtruncatemacro{\portymax}{5*\m}
  \foreach \j in {1,...,\kports}{%
    \draw[port] ({4*\j-1},0) rectangle ({4*\j},\portymax);
    \node[portlabel,font=\tiny,inner sep=0.4pt] at ({4*\j-0.5},{\portymax+0.2}) {$\Pi_{\j}$};
  }
  \draw[bound] (0,0) rectangle (\xmax,\ymax);
  \draw[rectU] (1,1) rectangle ({8*\m+2},2);
  \node[rectlabel,font=\tiny,inner sep=0.2pt] at ({(1+8*\m+2)/2},1.5) {$U$};
  \draw[rectP] (1,3) rectangle ({4*\m+2},4);
  \node[rectlabel,font=\tiny,inner sep=0.2pt] at ({(1+4*\m+2)/2},3.5) {$P$};
  \draw[rectQ] ({4*\m+1},3) rectangle ({8*\m+2},4);
  \node[rectlabel,font=\tiny,inner sep=0.2pt] at ({(4*\m+1+8*\m+2)/2},3.5) {$Q$};
  \foreach \i in {1,...,\last}{%
    \draw[rectE] ({4*\i+1},{5*\i+1}) rectangle ({4*\m+4*\i+2},{5*\i+2});
    \node[rectlabel,font=\tiny,inner sep=0.2pt] at ({(4*\i+1+4*\m+4*\i+2)/2},{5*\i+1.5}) {$E_{\i}$};
    \draw[rectA] ({4*\i-3},{5*\i+1}) rectangle ({4*\i+2},{5*\i+4});
    \node[rectlabel,font=\tiny,inner sep=0.2pt] at ({(4*\i-3+4*\i+2)/2},{5*\i+2.5}) {$A_{\i}$};
    \draw[rectD] ({4*\m+4*\i+1},{5*\i+1}) rectangle ({8*\m+2},{5*\i+4});
    \node[rectlabel,font=\tiny,inner sep=0.2pt] at ({(4*\m+4*\i+1+8*\m+2)/2},{5*\i+2.5}) {$D_{\i}$};
    \draw[rectB] ({4*\i+1},{5*\i+3}) rectangle ({4*\m+2},{5*\i+4});
    \node[rectlabel,font=\tiny,inner sep=0.2pt] at ({(4*\i+1+4*\m+2)/2},{5*\i+3.5}) {$B_{\i}$};
    \draw[rectC] ({4*\m+1},{5*\i+3}) rectangle ({4*\m+4*\i+2},{5*\i+4});
    \node[rectlabel,font=\tiny,inner sep=0.2pt] at ({(4*\m+1+4*\m+4*\i+2)/2},{5*\i+3.5}) {$C_{\i}$};
  }
  \draw[bound] (0,0) rectangle (\xmax,\ymax);
\end{tikzpicture}%
}

\newcommand{\DrawBasic}{%
\begin{tikzpicture}[x=0.7cm,y=0.7cm, line join=round]
\tikzset{
  translucent/.style={fill opacity=0.25, draw opacity=1},
  port/.style={draw=gray!60!black, fill=gray!25, line width=0.45pt},
  rectA/.style={draw=blue, fill=blue!35, translucent, line width=1.1pt},
  rectB/.style={draw=orange, fill=orange!35, translucent, line width=1.1pt},
  rectC/.style={draw=green!60!black, fill=green!35, translucent, line width=1.1pt},
  rectD/.style={draw=red!55!purple, fill=red!35, translucent, line width=1.1pt},
  rectE/.style={draw=red!70!black, fill=red!35, translucent, line width=1.1pt},
  rectPQ/.style={draw=teal!70!black, fill=teal!35, translucent, line width=1.1pt},
  rectU/.style={draw=gray!70!black, fill=gray!35, translucent, line width=1.1pt},
  lab/.style={font=\large},
  portlab/.style={font=\large}
}

  \draw[line width=1.4pt] (0,0) rectangle (9.5,6.4);

  \foreach \x/\name in {1.5/1,3.5/2,5.5/3,7.5/4} {
   \draw[port] (\x,0.08) rectangle ++(0.5,6.29);
    \node[portlab] at (\x+0.25,6.78) {$\Pi_{\name}$};
  }

  \draw[rectA] (0.5,3.85) rectangle (3.0,5.8);
  \node[lab] at (1.65,4.78) {$A$};

  \draw[rectB] (2.5,5.1) rectangle (5.2,5.8);
  \node[lab] at (3.75,5.43) {$B$};

  \draw[rectC] (4.5,5.1) rectangle (7.0,5.8);
  \node[lab] at (5.75,5.43) {$C$};

  \draw[rectD] (6.5,3.85) rectangle (9.0,5.8);
  \node[lab] at (7.75,4.78) {$D$};

  \draw[rectE] (2.5,3.85) rectangle (7.0,4.45);
  \node[lab] at (4.75,4.15) {$E$};

  \draw[rectPQ] (0.5,1.90) rectangle (5.0,2.55);
  \draw[rectPQ] (4.5,1.90) rectangle (9.0,2.55);
  \node[lab] at (2.75,2.22) {$P$};
  \node[lab] at (6.75,2.22) {$Q$};

  \draw[rectU] (0.5,0.65) rectangle (9.0,1.30);
  \node[lab] at (4.75,0.98) {$U$};
\end{tikzpicture}
}
\newcommand{\DrawGridSchematic}{%
\begin{tikzpicture}[x=1mm,y=1mm]
  \foreach \j/\x in {1/24,2/49,3/74,4/99}{%
    \draw[vcontainer] ({\x-3},3) rectangle ({\x+3},96);
  }
  \foreach \i/\y in {1/84,2/61,3/38,4/15}{%
    \draw[hcontainer] (10,{\y-3}) rectangle (113,{\y+3});
  }
  \foreach \j/\x in {1/24,2/49,3/74,4/99}{%
    \node[gridlabel,above] at (\x,96) {$V_{\j}$};
  }
  \foreach \i/\y in {1/84,2/61,3/38,4/15}{%
    \node[gridlabel,left] at (10,\y) {$H_{\i}$};
  }
  \foreach \i/\y in {1/84,2/61,3/38,4/15}{%
    \foreach \j/\x in {1/24,2/49,3/74,4/99}{%
      \draw[hport] ({\x-8},{\y-3}) rectangle ({\x+8},{\y+3});
      \draw[vport] ({\x-3},{\y-8}) rectangle ({\x+3},{\y+8});
    }
  }
\end{tikzpicture}%
}

\title{Counterexamples to Wegner's Conjecture for Rectangles}
\author{
Deepak Ajwani\thanks{University College Dublin, Ireland} \and
Rishikesh Gajjala\thanks{Center for Quantum and Topological Systems, New York University Abu Dhabi, UAE.}
\and
Rajiv Raman\thanks{Indraprastha Institute of Information Technology Delhi, India.}
\and
Saurabh Ray\thanks{New York University Abu Dhabi, UAE.}
}
\date{}

\begin{document}
\maketitle

\begin{abstract}
Wegner conjectured in 1965 that every finite family $\mathcal R$ of axis-parallel rectangles satisfies $\tau(\mathcal R)\le 2\nu(\mathcal R)-1$, where $\tau(\mathcal R)$ is the minimum number of piercing points and $\nu(\mathcal R)$ is the maximum size of a pairwise-disjoint subfamily.  We disprove the conjecture by an explicit triangle-free family of $64$ rectangles with $\nu=16$ and $\tau\ge 32$.

More generally, for every $\varepsilon>0$, we construct triangle-free rectangle families for which the standard clique-LP relaxation for maximum independent set of rectangles has integrality gap at least $5/2-\varepsilon$.  The same families satisfy $\tau(\mathcal R)\ge (5/2-\varepsilon)\nu(\mathcal R)$.  We also prove that, on triangle-free rectangle families, this LP has gap at most $3$.  Our approach gives an example with axis-parallel segments instead of rectangles with integrality gap tending to $2$. We also give a relatively small $4092$-rectangle triangle-free family with chromatic number $6$ improving the construction
of Asplund and Gr\"unbaum (On a coloring problem, Mathematica Scandinavica, 1960) that required more than $10^8$ rectangles.
\end{abstract}

\section{Introduction}
For a finite family $\Rcal$ of axis-parallel rectangles in the plane, let $\nu(\Rcal)$ be the largest size of a pairwise-disjoint subfamily and let $\tau(\Rcal)$ be the minimum number of points that pierce all rectangles.  Wegner's conjecture \cite{wegner1965} asserts the inequality
\begin{equation}\label{eq:wegner}
  \tau(\Rcal)\le 2\nu(\Rcal)-1.
\end{equation}
It is clear that $\tau(\Rcal)\ge \nu(\Rcal)$.
In one dimension, i.e., for intervals, the analogous packing-piercing relation is exact: $\tau=\nu$. In two dimensions, i.e., for rectangles, there is a gap between $\tau$ and $\nu$. For example, for a family of rectangles whose intersection graph is a 5-cycle, $\tau=3$, while $\nu = 2$.

Bounding transversal numbers in terms of packing numbers is well-studied in Discrete Geometry and Graph Theory \cite{Matousek2002, Berge1989}.  
Gy\'arf\'as and Lehel \cite{gyarfasLehel1985} developed and surveyed such covering and coloring questions for geometric set systems related to intervals, including box graphs and multiple-interval systems, and asked for linear packing-piercing bounds in several of these settings.  K\'arolyi \cite{karolyi1991} improved earlier general bounds for piercing rectangle families with no $k+1$ pairwise-disjoint members to an $O(k\log k)$ bound, and K\'arolyi and Tardos \cite{karolyiTardos1996} further connected the rectangle problem with point covers for multiple intervals.  Fon-Der-Flaass and Kostochka \cite{fonderflaassKostochka1993} studied the corresponding problem for boxes in higher dimension, giving dimension-dependent transversal bounds. 
The authors also exhibit a planar rectangle family with $\tau=5$ and $\nu=3$, showing that the ratio $\tau/\nu$ can be as large as $5/3$.

Aronov, Ezra, and Sharir \cite{aronovEzraSharir2010} proved $\varepsilon$-nets of size $O((1/\varepsilon)\log\log(1/\varepsilon))$ for the Hitting Set problem for axis-parallel rectangles along with related bounds for boxes.
Their results hold even when the hitting set must be chosen from a discrete set of points. 
This implies an integrality gap of $O(\log\log \OPT)$ for the LP-relaxation of this problem  and also yields a polynomial time $O(\log\log \OPT)$-approximation algorithm.
Pach and Tardos~\cite{PachTardos2013} showed a lower bound of $\Omega((1/\varepsilon)\log\log(1/\varepsilon))$ for discrete $\varepsilon$-nets, and hence a lower bound of $\Omega(\log\log \OPT)$ for the natural LP-relaxation of the Hitting Set problem in the discrete setting.
This line of work also gives the best general upper bound currently known, $\tau=O(\nu(\log\log \nu)^2)$. 
Correa, Feuilloley, P\'erez-Lantero, and Soto \cite{correaEtAl2015} studied independent and hitting sets for the structured case of rectangles intersecting a diagonal line. Among other results, they report a construction for which the ratio $\tau/\nu$ is arbitrarily close to $2$. However, their example does not violate
Wegner's conjecture. 

 Chudnovsky, Spirkl, and Zerbib \cite{chudnovskySpirklZerbib2018} obtained piercing results for special families of axis-parallel boxes.  Chen and Dumitrescu \cite{chenDumitrescu2020} showed that Wegner's proposed bound, if true, could not be strengthened at $\nu=4$ by constructing examples with $\nu=4$ and $\tau=7$.  Tomon \cite{tomon2023} showed that for boxes in dimension $d \ge 3$, $\tau$ is not bounded by any linear function of $\nu$. 
 In particular, he showed families with
 $\tau=\Omega\left(\nu\left({\log\nu}/{\log\log\nu}\right)^{d-2}\right)$.

A very closely related problem is the maximum independent set of rectangles (MISR): given a rectangle family, find the largest pairwise-disjoint subfamily.  Fowler, Paterson, and Tanimoto \cite{fowlerPatersonTanimoto1981} proved that this problem and the corresponding covering problem are NP-complete. 
Chalermsook and Chuzhoy \cite{chalermsookChuzhoy2009} gave a polynomial time $O(\log\log n)$-approximation by rounding the LP relaxation and constructed instances with asymptotic integrality gap $3/2$ for the standard clique-LP relaxation. Chalermsook~\cite{Chalermsook11} also gave a simpler $O(\log\log n)$-approximation. Adamaszek and Wiese \cite{adamaszekWiese2013} gave a quasi-polynomial-time approximation scheme (QPTAS) for the weighted problem, and Chuzhoy and Ene \cite{chuzhoyEne2016}  obtained an asymptotically faster QPTAS. Mitchell \cite{mitchell2021} gave the first polynomial-time constant-factor approximation for unweighted MISR that was improved by G\'alvez, Khan, Mari, M\"omke, Pittu, and Wiese \cite{galvezEtAl2022} to a $3$-approximation\footnote{There is an unpublished version on \href{https://arxiv.org/abs/2106.00623}{arXiv}  which claims to improve the approximation factor to $2+\varepsilon$ \cite{GalvezKhanMariMomkeReddyWiese2021}.}. For the special case of axis-parallel segments, Caoduro, Cslovjecsek, Pilipczuk, and W\k{e}grzycki \cite{caoduroEtAl2023} proved essentially tight bounds showing that the integrality gap of the LP relaxation approaches $2$.

Our constructions are related to and inspired by Asplund and Gr\"unbaum \cite{AsplundGrunbaum1960} who proved that triangle-free rectangle intersection graphs are $6$-colorable and gave a tight example. They also conjecture that for any rectangle intersection graph, the chromatic number $\chi$ is bounded above by $O(\omega)$ where $\omega$ is the size of a maximum clique in the graph. They also show that $\chi = O(\omega^2)$. Chalermsook and Walczak~\cite{ChalermsookW21} improved the bound to $O(\omega\log\omega)$. 
The study on graph classes where the chromatic number is upper bounded by a function of the clique number was initiated by Gy\'arf\'as~\cite{gyarfas1987problems} (see also~\cite{scott2020survey} for a recent survey on progress in this area).
Such graphs are called $\chi$-bounded. 
Prominently, in this direction, Erd\H{o}s conjectured that the intersection graphs of segments in the plane are $\chi$-bounded. This was
disproved by Pawlik, Kozik, Krawczyk, Lasoń, Micek, Trotter and Walczak~\cite{PAWLIK20146} showing that there exist triangle-free intersection graphs of segments with chromatic number
$\Omega(\log\log n)$, where $n$ is the number of segments.

\subsection{Our contribution}

Our first contribution is an explicit counterexample to Wegner's conjecture.  The instance has $64$ rectangles arranged as a $4\times4$ grid of eight-rectangle gadgets.  It is triangle-free, has $\nu=16$, and satisfies $\tau\ge 32$; hence
\[
  \tau\ge 32>31=2\nu-1.
\]
The same family has clique-LP value at least $32$ and integral optimum $16$, so it already gives an integrality gap of $2$.  This shows that the ``-1'' in Wegner's conjecture does not hold. 
Our main result is stronger, namely that even the factor ``2'' is incorrect.

\begin{restatable}{theorem}{mainGapTheorem}\label{thm:main-gap}
For every $\varepsilon>0$, there is a triangle-free family of axis-parallel rectangles for which the standard LP relaxation with clique constraints for the MISR problem has integrality gap at least $5/2-\varepsilon$.  As a consequence, the same family satisfies $\tau(\Rcal)\ge (5/2-\varepsilon)\nu(\Rcal)$.
\end{restatable}

This improves the previous rectangle lower bound for the integrality gap from $2 - o(1)$~\cite{correaEtAl2015, caoduroEtAl2023} to $5/2-o(1)$.  Since the constructed families are triangle-free, no point can pierce three rectangles.

\medskip
The main gadget we use is called a {\em package}.
This is a modification of the notion of a ``filter-bed'' used by Asplund and Gr\"unbaum~\cite{AsplundGrunbaum1960}. The only change is that a certain coloring property satisfied by filter-beds relevant for the coloring problem studied in their paper has been replaced by a property more suitable for bounding the sizes of independent sets. 

We complement Theorem~\ref{thm:main-gap} by showing that for triangle-free rectangle families the LP solution can be rounded within a factor 3.

\begin{restatable}[Weighted LP-relative 3-approximation]{theorem}{lpRelativeThreeTheorem}\label{thm:lp-relative-three}
Let \(\mathcal R\) be a triangle-free family of axis-parallel rectangles, and let
\(w:\mathcal R\to \mathbb R_{\ge 0}\) be nonnegative weights.  There is an
LP-rounding algorithm that returns an independent set \(I\subseteq \mathcal R\)
with
\[
    w(I)\ge \frac{\operatorname{OPT}_{LP}^w(\mathcal R)}{3},
\]
where \(\operatorname{OPT}_{LP}^w(\mathcal R)\) is the optimum value of the
weighted clique-LP relaxation.  In particular, the weighted integrality gap of
the clique-LP on triangle-free rectangle instances is at most \(3\).
\end{restatable}

Thus, for triangle-free rectangle families, integrality gap of the standard clique-LP lies between $5/2$ and $3$.  In particular, any example with gap larger than $3$ cannot be triangle-free.
The package viewpoint also gives a concise alternate construction yielding the known integrality gap of nearly $2$ for axis-parallel segments~\cite{caoduroEtAl2023}.

\begin{restatable}[Segments]{theorem}{segmentGapTheorem}\label{thm:segment-gap}
For every integer $k\ge2$, there is a triangle-free family of $2k(2k-1)$ axis-parallel segments whose clique-LP has an integrality gap of at least $2 -1/k$.
\end{restatable}

Finally, we show a filter-bed construction which can replace the filter-bed in the construction of Asplund and Gr\"unbaum~\cite{AsplundGrunbaum1960} improving the size of the constructed family significantly. Asplund and Gr\"unbaum described a construction using $12\cdot(8^{8} + 8^7 + \cdots + 8^0) \approx 2.3 \times 10^8$ rectangles and mentioned that they have a ``complicated'' example with ``only'' 50000 rectangles, and explicitly ask if a smaller family exists. Our filter-bed plugged into their construction yields a family of size $12 \cdot (4^{4} + 4^3 + \cdots + 4^0)=4092$.

\begin{restatable}[Chromatic number]{theorem}{chromaticSixTheorem}\label{thm:chromatic-six}
There is a triangle-free family of $4092$ rectangles with chromatic number $6$.
\end{restatable}

\paragraph{Paper organization.}
Section~\ref{sec:prelim} fixes notation, introduces the clique-LP, and defines packages.  Section~\ref{sec:eight-four} presents the $64$-rectangle counterexample to Wegner's conjecture.  Section~\ref{sec:general-package}
proves Theorem~\ref{thm:main-gap} and 
the LP-relative $3$-approximation for triangle-free rectangle families.  Section~\ref{sec:segments} gives the construction of a triangle-free segment family with integrality gap $2-\varepsilon$ for any $\varepsilon>0$.  Section~\ref{sec:coloring} gives the construction of a triangle-free rectangle family with 4092 rectangles and chromatic number $6$.  Section~\ref{sec:conclusion} lists open problems.

\section{Preliminaries}\label{sec:prelim}

All rectangles in this paper are closed axis-parallel rectangles.  For a rectangle $R$, write $I_x(R)$ and $I_y(R)$ for its projections on the $x$- and $y$-axes.  Rectangle $R_1$ \emph{crosses} rectangle $R_2$ \emph{horizontally} if $I_x(R_2)\subsetneq I_x(R_1)$ and $I_y(R_1)\subsetneq I_y(R_2)$.  It crosses $R_2$ \emph{vertically} if $I_x(R_1)\subsetneq I_x(R_2)$ and $I_y(R_2)\subsetneq I_y(R_1)$.  We say simply that two rectangles \emph{cross} if one of these two alternatives holds.

The intersection graph $G(\Rcal)$ of a rectangle family $\Rcal$ has vertex set $\Rcal$ and an edge between two rectangles iff they intersect.  An independent set in $G(\Rcal)$ is a pairwise-disjoint subfamily of $\Rcal$. Note that $\nu(\Rcal)=\alpha(G(\Rcal))$.

Axis-parallel rectangles satisfy the Helly property: if a finite subfamily is pairwise intersecting, then all rectangles in the subfamily share a common point.  Therefore cliques in $G(\Rcal)$ are exactly subfamilies pierced by one point, and the clique constraints for MISR are equivalent to point-depth constraints.  We use the LP
\begin{equation}\label{eq:lp}
\begin{array}{lll}
\text{maximize} & \displaystyle\sum_{R\in\Rcal}x_R,\smallskip\\
\text{subject to} & \displaystyle\sum_{R\in C}x_R\le 1 & \text{for every clique }C\text{ of }G(\Rcal),\smallskip\\
&0\le x_R\le 1 & \text{for every }R\in\Rcal.
\end{array}
\end{equation}
This is the standard LP relaxation with clique constraints.  For $N$ rectangles, it has an $O(N^2)$ size explicit representation. We need one constraint for each cell in the arrangement of the rectangles which can be efficiently enumerated by a sweep-line algorithm. 

If $G(\Rcal)$ is triangle-free, every clique has size at most two, and hence $x_R=1/2$ for every $R$ is feasible for the LP.  Thus $\LP(\Rcal)\ge |\Rcal|/2$ whenever $G(\Rcal)$ is triangle-free.  The same observation gives a piercing lower bound: in a triangle-free rectangle family, no point lies in three rectangles, so every piercing point hits at most two rectangles and therefore  $\tau(\Rcal)\ge |\Rcal|/2$.

We now define the main gadget in our constructions called a {\em package}. It is inspired by the ``filter-bed'' construction of Asplund and Gr\"unbaum~\cite{AsplundGrunbaum1960}.

\begin{definition}[Horizontal packages]\label{def:package}
A \emph{horizontal $(n,k)$-package} consists of a bounding rectangle $\Omega$, called its \emph{container}, a family $\Rcal$ of $n$ rectangles contained in $\Omega$, and $k$ pairwise-disjoint vertical slabs $\Pi_1,\ldots,\Pi_k\subseteq\Omega$, called \emph{ports}.  A rectangle crosses port $\Pi_j$ if it horizontally crosses $\Pi_j$.
For $R\in\Rcal$, let $\sigma(R)=\{j:R\text{ crosses }\Pi_j\}$.  The following five properties are required.
\begin{enumerate}[label=(P\arabic*), start=0]
\item For every $R\in\Rcal$ and every $j\in[k]$, either $R$ crosses $\Pi_j$ or $R\cap\Pi_j=\emptyset$.\label{prop:p0}
\item Every rectangle crosses at least one port.
\item For each port $\Pi_j$, the rectangles crossing $\Pi_j$ are pairwise disjoint.
\item The intersection graph $G(\Rcal)$ is triangle-free.
\item For every $Y\subseteq[k]$, if $\Rcal[Y]=\{R\in\Rcal:\sigma(R)\subseteq Y\}$, then $\alpha(\Rcal[Y])\le |Y|$.
\end{enumerate}
\end{definition}

\noindent
If a rectangle $R$ in the package crosses a port $\Pi$ of the package, we say that $R$ {\em uses} the port $\Pi$. We call the set of ports used by a rectangle its {\em port set}. 

\medskip\noindent
A \emph{vertical $(n,k)$-package} is obtained by a $90^\circ$ clockwise rotation of a horizontal package. In this case the ports are horizontal slabs in the container, and rectangles in the package cross the ports vertically.

\section{A 64-rectangle counterexample}\label{sec:eight-four}
In this section, we construct a relatively small collection of rectangles which violate Wegner's conjecture. We generalize the construction in the next section where we also prove things more formally. The main gadget we use is the horizontal $(8,4)$-package shown in 
Figure~\ref{fig:package-8-4} and its vertical version obtained by a clockwise $90^\circ$ rotation.
The shaded vertical slabs in Figure~\ref{fig:package-8-4} are the four ports. The formal coordinate construction of the package appears as the case $m=2$ of the general construction in Section~\ref{sec:general-package}. This is a simplified version of the ``filter-bed'' construction of Asplund and Gr\"{u}nbaum~\cite{AsplundGrunbaum1960}. 

\begin{figure}[htbp]
\centering
\small
\begin{minipage}[t]{0.48\textwidth}
\centering
\resizebox{\linewidth}{!}{\DrawBasic}
\caption{A horizontal $(8,4)$-package.\\ The intersection graph is $C_5\dotcup K_2\dotcup K_1$.}
\label{fig:package-8-4}
\end{minipage}\hfill
\begin{minipage}[t]{0.48\textwidth}
\centering
\resizebox{\linewidth}{!}{\DrawGridSchematic}
\caption{The $4\times4$ grid composition.\\ Shaded rectangles denote the ports.}
\label{fig:grid-schematic}
\end{minipage}
\end{figure}
It is easy to visually check that properties (P0), (P1), (P2) and (P3) in Definition~\ref{def:package} are satisfied. 
To verify Property (P4), we do a case analysis based on the size of the set $Y$ of {\em allowed} ports. If no port is allowed, i.e., $|Y| = 0$, clearly $|\Rcal[Y]|=0$ as well. If $|Y|=1$, depending on which port is allowed, $\Rcal[Y]$ contains exactly one of $A, B, C$ or $D$ and thus has size $1$. 
Suppose now that $|Y|=2$. If  $Y = \{\Pi_1,\Pi_2\}$ the available rectangles are $A,B,P$, with $A$ intersecting $B$. The case for $Y = \{\Pi_3,\Pi_4\}$ is symmetric. For $Y=\{\Pi_2,\Pi_3\}$ the available rectangles are $B, C$ and $E$. However, $B$ and $C$ are adjacent, and hence the independent set size
is $2$. For the remaining three two-port choices, the set of available rectangles has size $2$ and therefore has independence number at most $2$. 
Suppose now that three ports are allowed i.e., exactly one port is blocked. If $\Pi_1$ or $\Pi_4$ are blocked, the intersection graph of the available rectangles consists of a four-vertex path along with an isolated vertex - which has an independence number of 3. If $\Pi_2$ or $\Pi_3$ are blocked, the intersection graph has four vertices of which two are adjacent - which also has an independence number of $3$. 
Finally, if $|Y| = 4$ i.e., all four ports allowed, all rectangles in the package are available and their intersection graph is $C_5\dotcup K_2\dotcup K_1$, whose independence number is $2+1+1=4$.  Thus the largest independent set using only $t$ of the four ports has size at most $t$.

We now construct a grid using scaled and translated copies of this package.  We take four horizontal copies $H_1,\ldots,H_4$ and four vertical rotated copies $V_1,\ldots,V_4$ and place them so that, for every pair $(i,j) \in \{1, 2, 3, 4\}^2$,  port $j$ of $H_i$ horizontally crosses the container of $V_j$ and port $i$ of $V_j$ vertically crosses the container of $H_i$.  Figure~\ref{fig:grid-schematic} shows the construction. For simplicity, we only show the containers of the packages and not the rectangles inside them.

The resulting family has $4\cdot8+4\cdot8=64$ rectangles. 
Note that for each pair $(i,j)$, all rectangles using port $j$ of $H_i$ intersect all rectangles using port $i$ of $V_j$ and these are the only intersections between rectangles from distinct packages. Property (P2) ensures that the intersection graph of the family of 64 rectangles is triangle-free which means that for this family $\tau \ge 64/2 = 32$. Next, note that for any pair $(i,j)$ any independent set of the rectangles cannot contain a rectangle using port $j$ in $H_i$ as well as a rectangle using port $i$ in $V_j$ as any two such rectangles intersect. Thus, an independent set can either use port $j$ of $H_i$ or port $i$ of $V_j$ but not both. Thus, over all 8 packages, the total number of ports used is at most 16. By property (P4) of packages, for this family of rectangles $\nu \le 16$. Conversely, each horizontal $(8,4)$-package contains the independent set $\{U,P,A,C\}$ of size $4$; taking such a set in each of the four horizontal packages (and no vertical-package rectangles) gives an independent set of size $16$, since distinct horizontal packages have disjoint containers. Hence $\nu=16$.

Thus we have $\tau > 2\nu -1$ contradicting Wegner's conjecture. 
The same instance also has LP value at least $64/2=32$ and integral optimum of $16$, so its integrality gap is 2.
A more formal version of the above intuitive proof appears in 
Theorem~\ref{thm:grid}.

\section{Instances with larger integrality gap}\label{sec:general-package}

We now generalize the idea in Section~\ref{sec:eight-four} so that the integrality gap of the LP approaches $5/2$.  The argument in the grid proof remains unchanged; what changes is the package construction so that the rectangles to port ratio approaches $5/2$. We first prove that this ratio is the integrality gap we obtain from the grid construction.

\begin{theorem}[Grid composition]\label{thm:grid}
Given a horizontal $(n,k)$-package, we can build a triangle-free family $\Fcal$ of $2nk$ rectangles with $\alpha(\Fcal) \le k^2$ and $\tau(\Fcal)\ge nk$. This implies that the integrality gap of the LP for the MISR problem with clique constraints is at least $n/k$.
\end{theorem}

\begin{proof}
Let $\Pcal$ be the given $(n,k)$-horizontal package and let $\Pcal^\perp$ be the vertical package obtained by a clockwise $90^\circ$ rotation of $\Pcal$.  Place $k$ horizontal copies $H_1,\ldots,H_k$ of $\Pcal$ with disjoint containers stacked vertically, and $k$ vertical copies $V_1,\ldots,V_k$ with disjoint containers placed side by side.  Scale and translate the copies so that, for every $(i,j) \in \{1, 2, \ldots, k\}^2$ , port $j$ of $H_i$ horizontally crosses port $i$ of $V_j$. 
Note that this forces the $x$-range of the container of $V_j$ to be a proper subset of the $x$-range of the port $j$ in $H_i$ and similarly the $y$-range of the container of $H_i$ must be a proper subset of the $y$-range of the port $i$ in $V_j$. By Property~\ref{prop:p0}, since each rectangle either crosses a port or is disjoint from it, this has two implications: i) all rectangles using port $j$ in $H_i$ intersect all rectangles using port $i$ in $V_j$ and ii) these are the only intersections among rectangles belonging to distinct packages - the remaining intersections are between rectangles within a package. 
Let $\Fcal$ denote the set of all rectangles in all the packages. Since there are $2k$ packages with $n$ rectangles each, $|\Fcal| = 2nk$. 
We next observe that $\Fcal$ is triangle-free. To see this note that by property (P3) of packages, the rectangles within a single package do not form a triangle. Since the distinct rectangles belonging to distinct horizontal packages do not intersect and similarly rectangles belonging to distinct vertical packages do not intersect, any triangle in the intersection graph must involve two rectangles from a horizontal package $H_i$ and one from a vertical package $V_j$ or vice-versa (two from $V_j$ and one from $H_i$). Assume that two of the rectangles belong to $H_i$, the other case being analogous. This leads to a contradiction since the only rectangles in $H_i$ that may intersect rectangles in $V_j$ are those that use the port $j$ of $H_i$ and rectangles in a package which use the same port must be pairwise disjoint by property (P2) of packages.

Next, we show that any independent set $I \subseteq \Fcal$ has size at most $k^2$.  For each horizontal package $H_i$, let $Y_i^H$ be the set of ports of $H_i$ used by rectangles of $I\cap H_i$.  By property (P4) of Definition~\ref{def:package}, $|I\cap H_i|\le |Y_i^H|$.  Similarly, for each vertical package $V_j$,  $|I\cap V_j|\le |Y_j^V|$ where $Y_j^V$ denotes the set of ports of $V_j$ used by rectangles of $I\cap V_j$.  For every pair $(i,j) \in \{1, 2, \ldots, k\}^2$, we cannot have both $j\in Y_i^H$ and $i\in Y_j^V$, any rectangle in $H_i$ using port $j$ intersects any rectangle in $V_j$ using port $i$.  
Thus, the total number of used ports over all packages is at most the number of cells: $\sum_i |Y_i^H|+\sum_j |Y_j^V|\le k^2$.  Consequently $|I|\le \sum_i |Y_i^H|+\sum_j |Y_j^V|\le k^2$.

Since $\Fcal$ is a triangle-free family, any point in the plane  pierces at most two rectangles, implying that $\tau(\Fcal)\ge |\Fcal|/2=nk$.  Also the uniform vector $x_R=1/2$ is feasible for the LP with clique constraints and has value $nk$, while the integral optimum, as argued above, is at most $k^2$.  The integrality gap is therefore at least $nk/k^2=n/k$.
\end{proof}

\noindent{\bf Package construction.} We now construct $(n,k)$-packages with $n/k$ approaching $5/2$ for large $k$.  The idea is to use disjoint vertically separated copies of five-cycles.  Each $C_5$ contributes five rectangles but only two units of ``independent-set capacity''.  The horizontal extents of the rectangles within each five-cycle are chosen carefully so that, even after many cycles are stacked together, the package still satisfies the port-restricted bound (P4).  All other package properties are nearly immediate from the construction.

Fix an integer $m\ge2$ and set $k=2m$.  The package has $k$ ports and $n=5m-2$ rectangles.  Its container is $\Omega_m=[0, 8m+3]\times[0,5m]$, and its ports are $\Pi_j=[4j-1,4j]\times[0,5m]$ for $j=1,\ldots,2m$.
For a nonempty consecutive interval $[a,b]\subseteq[2m]$, define $X[a,b]=[4a-3,4b+2]$.  A rectangle with horizontal span $X[a,b]$ crosses exactly the ports $\Pi_a,\Pi_{a+1},\ldots,\Pi_b$ and is separated by a positive gap from every other port. 
The package contains the following rectangles: three initial rectangles $U, P,Q$ along with five rectangles $A_i,B_i,C_i,D_i,E_i$ for each $i=1,\ldots,m-1$. The $x$- and $y$-ranges of the rectangles are shown in the table below. 

\begin{center}
\begin{tabular}{ccc}
\toprule
rectangle & $x$-range & $y$-range\\
\midrule
$U$ & $X[1,2m]$ & $[1,2]$\\
$P$ & $X[1,m]$ & $[3,4]$\\
$Q$ & $X[m+1,2m]$ & $[3,4]$\\
\midrule
$A_i$ & $X[i,i]$ & $[5i+1,5i+4]$\\
$B_i$ & $X[i+1,m]$ & $[5i+3,5i+4]$\\
$C_i$ & $X[m+1,m+i]$ & $[5i+3,5i+4]$\\
$D_i$ & $X[m+i+1,2m]$ & $[5i+1,5i+4]$\\
$E_i$ & $X[i+1,m+i]$ & $[5i+1,5i+2]$\\
\bottomrule
\end{tabular}
\end{center}
Note that $n=5(m-1) + 3$, $k = 2m$ and $n/k = 5/2 - 2/k$. The intersection graph of the rectangles has the following connected components each placed in a disjoint horizontal band: the isolated vertex $U$, the edge $P - Q$ and a $5$-cycle $A_i-B_i-C_i-D_i-E_i-A_i$ for each $i=1, \ldots, m-1$. This can be easily verified from the coordinates in the above table.  Figure~\ref{fig:package-18-8} shows the constructions for $m=4$. 

\begin{proposition}\label{prop:package-general}
For every $m\ge2$, the construction above is a horizontal $(5m-2,2m)$-package.  
\end{proposition}

\begin{proof}
Every rectangle with an $x$-range $X[a,b]$ crosses exactly $\Pi_a,\ldots, \Pi_b$ and is disjoint from the other ports; this gives \ref{prop:p0}.
Every listed port interval is nonempty, so every rectangle crosses at least one port satisfying property (P1).
Property (P3) is satisfied since the intersection graph is $(m-1)C_5\dotcup K_2\dotcup K_1$ which is triangle-free. 
Property (P2) is also easy to verify. 
Fix a port $\Pi_j$.  Rectangles belonging to different $5$-cycles lie in vertically disjoint bands.
Also $U$ and $P, Q$ lie in bands of their own disjoint from others. 
Inside one $5$-cycle, the only pairs of rectangles sharing a port are $(B_i,E_i)$ and $(C_i,E_i)$, and both pairs are vertically disjoint. Thus all rectangles crossing any fixed port are pairwise disjoint. 
It remains to check the property (P4). First take a consecutive block of ports $J=[r,s]\subseteq[2m]$, and let $\Rcal(J)$ be the rectangles whose port set is contained in $J$.
If $s<m$, only singleton rectangles $A_i$ can lie in $J$, so $\alpha(\Rcal(J))\le |J|$.

Assume $m\le s< 2m$.  If $r>m+1$, $J$ does not contain the port set of any rectangle.  Otherwise, the part of the $i$-{th} cycle contained in $J$ has independence number at most $\1_{\{i\ge r-1\}}+\1_{\{i\le s-m\}}$.  The first term accounts for the possible left side $(A_i,B_i)$, and the second for the possible middle-right side $(C_i,E_i)$; no $D_i$ is present because $s<2m$.  Summing over $i$ and adding the possible rectangle $P$ when $r=1$, we get $\alpha(\Rcal(J))\le (m-1)+(s-m)+1=s$ if $r=1$, and $\alpha(\Rcal(J))\le (m-r+1)+(s-m)=s-r+1$ if $2\le r\le m+1$.  Thus $\alpha(\Rcal(J))\le |J|$.

Finally let $s=2m$.  If $r>m+1$, only suffix rectangles $D_i$ can appear, so there are at most $2m-r+1=|J|$ of them.  If $r\le m+1$, then the $i^{th}$ cycle contributes at most $1+\1_{\{i\ge r-1\}}$: one unit for the right side $(C_i,D_i)$, and a second unit only when the left side $(A_i,B_i,E_i)$ is also available.  The rectangles in $\{U, P, Q\}$ contribute 2 when $r=1$ since the independence number of $\{U, P, Q\}$ is 2, and contribute 1 when $2\le r\le m+1$, because only $Q$ is available.  Hence $\alpha(\Rcal(J))\le (m-1)+(m-1)+2=2m$ if $r=1$, and $\alpha(\Rcal(J))\le (m-1)+(m-r+1)+1=2m-r+1$ if $2\le r\le m+1$.  Again $\alpha(\Rcal(J))\le |J|$.

For an arbitrary set of ports $Y\subseteq[2m]$, decompose $Y$ into maximal consecutive blocks $J_1,\ldots,J_q$.  Any rectangle whose port set is contained in $Y$ has port set contained in one of these blocks.  Rectangles assigned to different blocks have disjoint $x$-projections, because a missing port separates the blocks.  Therefore the independence numbers add, and $\alpha(\Rcal[Y])=\sum_h\alpha(\Rcal(J_h))\le \sum_h |J_h|=|Y|$.  This proves (P4).
\end{proof}

\begin{figure}[htbp]
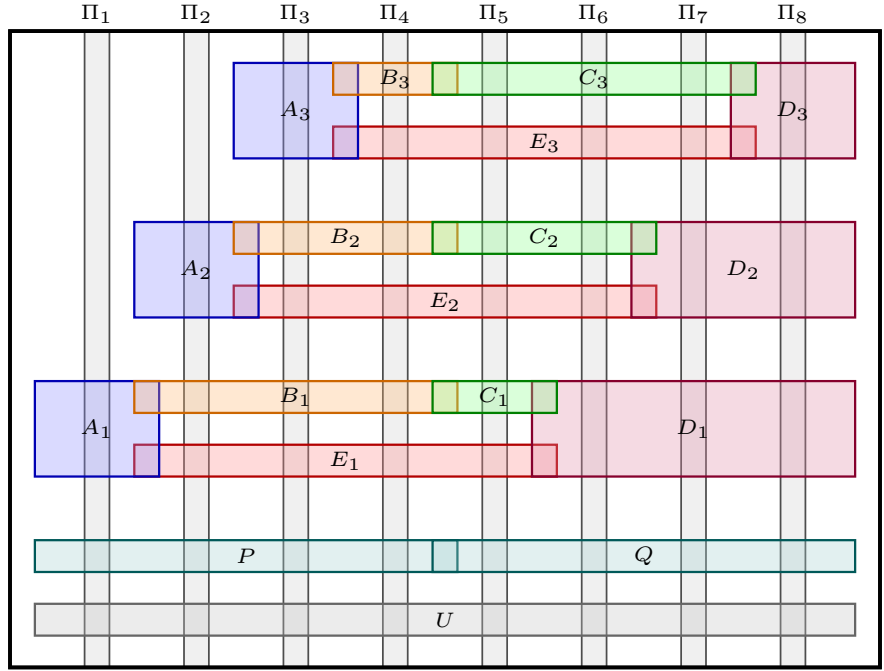

\centering
\small
\resizebox{0.70\textwidth}{!}{\DrawPackage{4}}
\caption{The case $m=4$ of the general construction: the horizontal $(18,8)$-package.}
\label{fig:package-18-8}
\end{figure}

Combining Theorem~\ref{thm:grid} with Proposition~\ref{prop:package-general}, and using $n=5m-2$ and $k=2m$, we get an integrality gap of $n/k=5/2-1/m$.  This proves Theorem~\ref{thm:main-gap}.

\mainGapTheorem*

\begin{corollary}
For every $\delta>0$, 
there is a triangle-free family of rectangles so that 
the integrality gap of the standard 
LP-relaxation for the vertex cover problem is at least $\frac{8}{5} - \delta$.    
\end{corollary}
\begin{proof}
The LP-relaxation of the vertex cover problem for any family $\Rcal$ of rectangles is: \[\min_{\{x_R: R\in\Rcal\}}\left\{\sum_{R\in\Rcal} x_R: x_A + x_{B}\ge 1,\forall \{A, B\}\in G(\Rcal), x_R\ge 0\right\}.\]
Setting $x_R = 1/2$ for all $R\in\Rcal$ yields a feasible solution with objective value $|\Rcal|/2$.
The proof of Theorem~\ref{thm:main-gap} constructs a triangle-free family of $n$ rectangles such that the maximum independent set has size at most $\frac{n}{2(2.5 - \varepsilon)} = \frac{n}{5 - 2\varepsilon}$
for any $\varepsilon>0$. In the same family the vertex cover therefore has size at least $\OPT := n - \frac{n}{5 - 2\varepsilon} = \frac{4 - 2\varepsilon}{5 - 2\varepsilon}n$, while the optimal solution to the vertex cover LP is at most $n/2$, since $x_R=\tfrac12$ for all $R$ is feasible.
The integrality gap is therefore  $\left(\frac{4 - 2\varepsilon}{5 - 2\varepsilon}n \right) / \left( n/2 \right) = \frac{8-4\varepsilon}{5 - 2\varepsilon} \ge  \frac{8}{5} - \delta$ for sufficiently small $\varepsilon$. 
\end{proof}

We next show that the integrality gap of triangle-free rectangle families cannot be more than $3$. 
In fact, we show that this upper bound holds even for the \emph{weighted problem} of computing a maximum weight independent set.
Thus to obtain significantly larger integrality gap, 
one needs to explore families with larger max-clique sizes. For a family of rectangles $\Rcal$, 
let $\OPT^w_{\LP}(\Rcal)$ denote the optimal solution value of the
LP relaxation~\eqref{eq:lp} for the maximum independent set problem on $\Rcal$.

\lpRelativeThreeTheorem*
\begin{proof}
Let \(x^*\) be an optimal solution to the \emph{weighted clique-LP}, that is obtained from LP-relaxation~(\ref{eq:lp}) by replacing the objective
function by $\sum_{R\in\mathcal{R}} w_Rx_R$. Then, 
\[
    \operatorname{OPT}_{LP}^w=\sum_{R\in\mathcal R} w_R x_R^* .
\]
Compute a \(6\)-coloring \(\mathcal R_1,\ldots,\mathcal R_6\) of the triangle-free
rectangle intersection graph - the proof of Asplund and Gr\"{u}nbaum that triangle-free rectangle intersection
graphs are $6$-colorable directly yields a polynomial time algorithm~\cite{AsplundGrunbaum1960}.  For each
color class, define
\[
    \mu_i=\sum_{R\in\mathcal R_i} w_R x_R^* .
\]
Since each color appears in exactly five of the \(\binom 62=15\) pairs of colors,
there is a pair \(i<j\) such that
\[
    \mu_i+\mu_j
    \ge
    \frac{1}{15}\sum_{i<j}(\mu_i+\mu_j)
    =
    \frac{5}{15}\sum_{i=1}^6 \mu_i
    =
    \frac{\operatorname{OPT}_{LP}^w}{3}.
\]
Let \(\mathcal R'=\mathcal R_i\cup \mathcal R_j\).  The graph \(G[\mathcal R']\)
is bipartite.  Consider the weighted stable-set LP on \(G[\mathcal R']\):
\[
    \max \sum_{R\in\mathcal R'} w_R y_R
    \quad\text{s.t.}\quad
    y_R+y_S\le 1 \text{ for every edge } RS\in E(G[\mathcal R']),
    \qquad 0\le y_R\le 1.
\]
The restriction of \(x^*\) to \(\mathcal R'\) is feasible and has value
\(\mu_i+\mu_j\).  Since \(G[\mathcal R']\) is bipartite, this LP is integral.
Therefore it has an integral optimum, which is an independent set of weight at
least \(\frac{1}{3}\operatorname{OPT}_{LP}^w\).
\end{proof}

\section{Integrality gap for axis-parallel segments}
\label{sec:segments}

We now give an alternate construction which yields the same integrality-gap lower bound for axis-parallel segments as in~\cite{caoduroEtAl2023}. It is interesting to note that the construction is similar to that in \cite{caoduroEtAl2023} but arrived at from a different perspective. 
By Theorem~\ref{thm:grid}, it suffices to construct a package of horizontal segments whose segments to ports ratio tends to $2$.
Fix $k\ge2$.  For each $i=1,\ldots,k-1$, take two horizontal segments
$L_i$ and $R_i$ with
$\sigma(L_i)=\{1,\ldots,i\}$ and $\sigma(R_i)=\{i+1,\ldots,k\}$.
Place $L_i$ and $R_i$ on the line $y = i$ so that they intersect in the gap between the ports 
$\Pi_i$ and $\Pi_{i+1}$. 
Finally add one more segment $U$ on the line $y=0$ with $\sigma(U)=\{1,\ldots,k\}$.
Thus the package has $2k-1$ segments.  Figure~\ref{fig:segment-package-k5}
shows the construction for $k=5$.

\begin{figure}[!htbp]
\centering
\begin{tikzpicture}[x=2.15cm,y=.62cm,line cap=round,font=\small]
  \draw[bound] (0,.58) rectangle (6.05,5.42);
  \foreach \j in {1,...,5} {
    \draw[fill=gray!15,draw=gray!60] (\j-.07,.58) rectangle (\j+.07,5.42);
    \node[above=1pt] at (\j,5.42) {$\Pi_{\j}$};
  }
  \draw[ultra thick,black!70] (.48,1) -- (5.52,1);
  \node[left] at (.48,1) {$U$};
  \foreach \i in {1,...,4} {
    \draw[ultra thick,blue!70!black] (.48,{\i+1}) -- (\i+.55,{\i+1});
    \draw[ultra thick,orange!80!black] (\i+.45,{\i+1}) -- (5.52,{\i+1});
    \node[left] at (.48,{\i+1}) {$L_{\i}$};
    \node[right] at (5.52,{\i+1}) {$R_{\i}$};
  }
  \draw[bound] (0,.58) rectangle (6.05,5.42);
\end{tikzpicture}
\vspace{0.5ex}
\caption{A $(9,5)$-package formed by horizontal segments.}
\label{fig:segment-package-k5}
\end{figure}
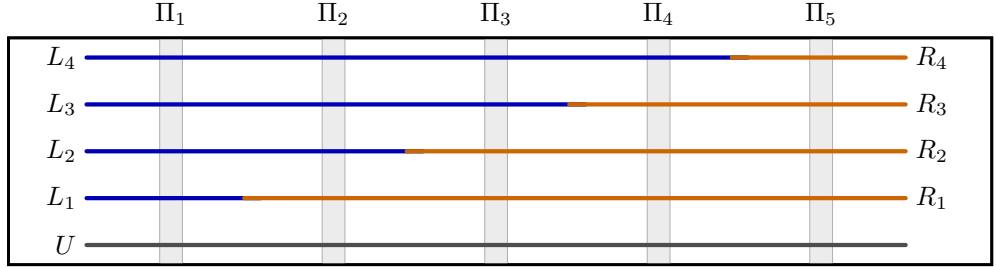
  
Properties (P0), (P1), (P2) and (P3) of packages hold by construction and are easy to check. 
Property (P4) holds trivially when $Y = \emptyset$. It is also easy to verify when $Y$ contains all $k$ ports: any  independent set may contain $U$ and at most one segment from each pair $(L_i, R_i)$ for $i = 1, \ldots, k-1$  and thus has size at most $k$. 
Assume now that $Y$ is non-empty but does not contain all ports. In this case $U \notin \Rcal[Y]$.
Let $r\ge 0$ be the largest number so that $Y$ contains the first $r$ ports and similarly let $s\ge 0$ be the largest number so that $Y$ contains the last $s$ ports. Since we assumed that $Y$ does not contain all ports, $r+s < k$ and therefore $|Y| \ge r + s$. In this case, note that $\Rcal[Y]$ contains exactly the first $r$ segments in $L_1, \ldots, L_{k-1}$ and the last $s$ segments in $R_1, \ldots, R_{k-1}$ and therefore has size exactly $r+s$. This implies that property (P4) holds: $\alpha(\Rcal[Y]) \le |\Rcal[Y]| = r+s \le |Y|$. 

Applying the grid composition of Theorem~\ref{thm:grid} to this package yields the following. 
\segmentGapTheorem*

For comparison, the family obtained in~\cite{caoduroEtAl2023} has size $4k^2$ with integrality gap $\frac{2k^2}{k^2 + 3k - 2} = 2 - 6/k + O(1/k^2)$. So, for the same integrality gap, our construction is slightly smaller. 

\section{Smaller triangle-free family with chromatic number 6}
\label{sec:coloring}

Asplund and Gr\"{u}nbaum~\cite{AsplundGrunbaum1960} proved that the chromatic number of any triangle-free family of axis-parallel rectangles is at most $6$. 
They also showed that this is tight by constructing a triangle-free family of $12 \cdot (8^{8+1} - 1)/(8-1) \approx 2.3 \times 10^8$ rectangles whose proper coloring requires 6 colors. 
The main gadget they used is a ``filter-bed'' which is simply a package (as defined in Definition~\ref{def:package}) with the property (P4) replaced by the following {\bf\em coloring property} (P4'): for any proper $5$-coloring of rectangles in the package, the rectangles using one of the ports have at least 3 different colors.
Asplund and Gr\"{u}nbaum~\cite{AsplundGrunbaum1960} use the filter-bed to construct a family of rectangles, with a combination of a hierarchical tree-like structure and a grid structure to force any proper coloring of the family to use at least $6$ colors.
Asplund and Gr\"{u}nbaum remark that they found another construction which was ``quite complicated'' but has only about 50000 rectangles. They however did not describe the construction. 
We improve on their bound by simply replacing the filter-bed in their construction by the new filter-bed shown in Figure~\ref{fig:filter-12-4}. This immediately reduces the number of rectangles in the construction to $12 \cdot (4^{4+1}-1)/(4-1) = 4092$. It only remains to prove that the filter-bed in Figure~\ref{fig:filter-12-4} satisfies the required properties. 

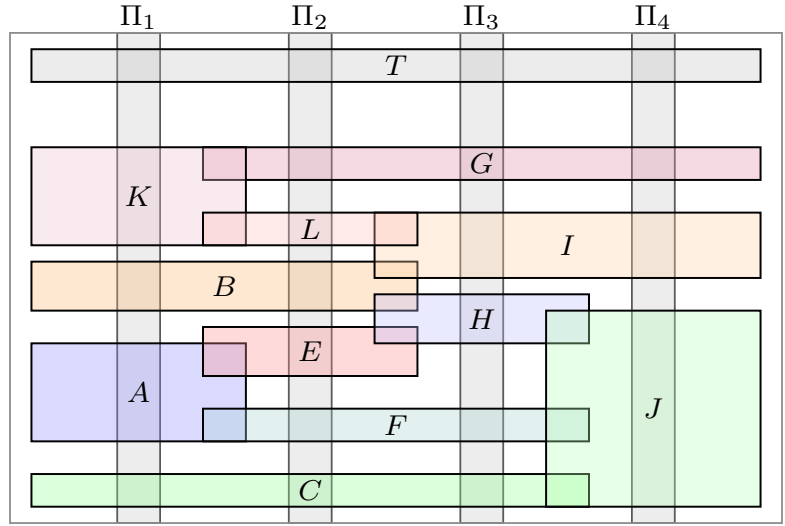
\begin{figure}[!htbp]
\centering
\small
\resizebox{0.62\textwidth}{!}{%
\begin{tikzpicture}[
  x=0.42cm,y=0.16cm,
  filterport/.style={draw=black!65,fill=black!25,fill opacity=0.28,line width=0.5pt},
  filterrect/.style={draw=black,fill=blue!35,fill opacity=0.32,line width=0.55pt},
  filterlabel/.style={font=\scriptsize,inner sep=0.6pt,text opacity=1}
]
  \foreach \j/\x in {1/3.5,2/7.5,3/11.5,4/15.5}{%
    \draw[filterport] ({\x-0.5},0) rectangle ({\x+0.5},30);
    \node[above,font=\scriptsize,inner sep=1pt] at (\x,30) {$\Pi_{\j}$};
  }
  \draw[black!45,line width=0.5pt] (0.5,0) rectangle (18.5,30);

  \draw[filterrect,fill=blue!45] (1,5) rectangle (6,11);
  \draw[filterrect,fill=orange!55] (1,13) rectangle (10,16);
  \draw[filterrect,fill=green!45] (1,1) rectangle (14,3);
  \draw[filterrect,fill=gray!45] (1,27) rectangle (18,29);
  \draw[filterrect,fill=red!45] (5,9) rectangle (10,12);
  \draw[filterrect,fill=teal!35] (5,5) rectangle (14,7);
  \draw[filterrect,fill=purple!45] (5,21) rectangle (18,23);
  \draw[filterrect,fill=blue!25] (9,11) rectangle (14,14);
  \draw[filterrect,fill=orange!35] (9,15) rectangle (18,19);
  \draw[filterrect,fill=green!30] (13,1) rectangle (18,13);
  \draw[filterrect,fill=purple!25] (1,17) rectangle (6,23);
  \draw[filterrect,fill=red!25] (5,17) rectangle (10,19);

  \foreach \name/\x/\y in {
    A/3.5/8,
    B/5.5/14.5,
    C/7.5/2,
    T/9.5/28,
    E/7.5/10.5,
    F/9.5/6,
    G/11.5/22,
    H/11.5/12.5,
    I/13.5/17,
    J/15.5/7,
    K/3.5/20,
    L/7.5/18
  }{
    \node[filterlabel] at (\x,\y) {$\name$};
  }
\end{tikzpicture}%
}
\caption{Modified filter-bed with 12 rectangles and 4 ports.}
\label{fig:filter-12-4}
\end{figure}

\begin{proposition}
The gadget shown in Figure~\ref{fig:filter-12-4} satisfies properties (P0), (P1), (P2), (P3) of packages as well as the coloring property (P4'). 
\end{proposition}
\begin{proof}
Properties (P0), (P1), (P2) and (P3) are immediate from visual inspection. 
It remains to prove the coloring property. 

Suppose, for a contradiction, that there is a proper $5$-coloring in which every port sees at most two colors. Since $T$ uses all four ports, the color of $T$ appears on every port; call this color $0$. This implies that for each port $j$, the colors seen by the port are contained in $S_j := \{0,s_j\}$ for some $s_j\ne 0$. 
If a rectangle uses the interval of ports $[a,b]$, then its color lies in $S_a\cap\ldots\cap S_b$. Thus such a rectangle has color $0$, unless $s_a=s_{a+1}=\ldots=s_b$, in which case the only possible nonzero color is this common value.

Let $P$ denote the event $s_1=s_2$, and let $Q$ denote the event $s_3=s_4$. We first show that $P$ and $Q$ cannot both hold. Consider the $5$-cycle $A-E-H-J-F-A$, and write $p=s_1=s_2$ and $q=s_3=s_4$. If $p=q$, then all vertices of this odd cycle have colors in $\{0,p\}$, impossible in a proper coloring. If $p\ne q$, then $F$ is forced to have color $0$. Hence its neighbors $A$ and $J$ are forced to have colors $p$ and $q$, respectively. Then $E$ is forced to have color $0$ because it is adjacent to $A$, and $H$ is forced to have color $0$ because it is adjacent to $J$. This contradicts the proper coloring of the edge $EH$.

Next, $P$ and $Q$ cannot both fail. If $P$ fails, then $B$ is forced to have color $0$. If $Q$ fails, then $I$ is forced to have color $0$. This contradicts the edge $BI$. Hence exactly one of $P$ and $Q$ holds.

Suppose first that $P$ holds and $Q$ fails. Since $Q$ fails, $I$ is forced to have color $0$. The edge $IL$ forces $L$ to have the nonzero color $s_2=s_1$. The edge $LK$ then forces $K$ to have color $0$. The edge $KG$ forces $G$ to have a nonzero color, but $G$ uses ports $2,3,4$ and $s_3\ne s_4$, so $G$ is forced to have color $0$, a contradiction.

Finally suppose that $Q$ holds and $P$ fails. Since $P$ fails, $B$ is forced to have color $0$. The edge $BH$ forces $H$ to have the nonzero color $s_3$. Also, $C$ uses ports $1,2,3$ and $s_1\ne s_2$, so $C$ is forced to have color $0$. The edge $CJ$ forces $J$ to have the nonzero color $s_4$. But $Q$ says $s_3=s_4$, contradicting the proper coloring of the edge $HJ$. This proves the proposition.
\end{proof}

As mentioned before, using the construction of Asplund and Gr\"{u}nbaum with their filter-bed replaced by the gadget in Figure~\ref{fig:filter-12-4}, we get the following. The proof below is an adaptation of the proof of  Asplund and Gr\"{u}nbaum \cite{AsplundGrunbaum1960} with only minor changes.

\chromaticSixTheorem*

\begin{proof}
Asplund and Gr\"unbaum~\cite{AsplundGrunbaum1960} showed that every triangle-free
rectangle family is $6$-colorable.  It is therefore enough to construct a
triangle-free family which is not $5$-colorable. 
The construction idea is exactly the same as the one used in~\cite{AsplundGrunbaum1960} with a slight change in the way we describe it. The family has a hierarchical structure mimicking a complete rooted $4$-ary tree $\mathcal{T}$ with depth $4$. At depth $d \in \{0, 1, 2, 3, 4\}$, there are $4^d$ nodes. The internal nodes in the tree correspond to horizontal filter-beds (defined analogously to horizontal packages) and the leaves correspond to vertical filter-beds (again defined analogously to vertical packages). 
For any node $N$ in the tree, denote the filter-bed corresponding to $N$ 
as $\filterbed(N)$ and its container as $\container(N)$.
We place these packages so that the following conditions are satisfied:
\begin{enumerate}
    \item The containers of all filter-beds corresponding to nodes in $\mathcal{T}$ at the same depth have the same $y$-range and disjoint $x$-ranges. 
    \item For any non-root node $C$ and its parent $P$, if $C$ is $i^{th}$ child of $P$ then 
    the $x$-range of $\container(C)$ is contained in the $x$-range of the $i^{th}$ port $\Pi_i$ of $\filterbed(P)$. 
    \item For any leaf $L$ in $\mathcal{T}$, the $y$-range of the port $\Pi_i$ of $\filterbed(L)$ contains the $y$-range of internal nodes at depth $i-1$ for all $i \in \{1,2,3,4\}$. 
\end{enumerate}
It is easy to see that there is a placement satisfying the above constraints (see~\cite{AsplundGrunbaum1960} for details).
First note that the total number of rectangles in the construction is $12$ times the number of nodes in the tree since each node corresponds to a filter-bed with 12 rectangles. This is $12 \cdot (1 + 4 + 4^2 + 4^3 + 4^4) = 4092$.

We first check that the constructed family is triangle-free. Each filter-bed is clearly triangle-free. Distinct filter-beds at the same depth have disjoint $x$-ranges, and horizontal filter-beds at different internal depths have disjoint $y$-ranges, so those pairs are disjoint. A leaf (vertical) filter-bed can meet a horizontal filter-bed only when the latter lies on that leaf's root-to-leaf path, and by (P0) such intersections occur only through the corresponding pair of ports. Since rectangles using a fixed port are pairwise disjoint (P2), no triangle can use rectangles from two different filter-beds. Hence the whole family is triangle-free.

We now claim that the intersection graph of the rectangles in the construction above is not $5$-colorable. 
For contradiction assume that there is a $5$-coloring.
By property $(P4')$ for every package in the construction, there is one port so that the rectangles using that port use at least $3$ colors - we call such a port {\em tricolored}. The main idea is to show that there exists a pair $(H, V)$ where $H$ is a horizontal filter-bed and $V$ is a vertical filter-bed such that a tricolored port $p$ of $H$ horizontally crosses a tricolored port $q$ of $V$. Then, all rectangles using $p$ in $H$ intersect all rectangles using $q$ in $V$ and thus they together must use $6$ colors, contradicting the $5$-colorability assumption.
Now, note that there exists a root to leaf path $v_1 - v_2 - v_3 - v_4 - v_5$ so that for any $i \in \{1,2,3,4\}$ if $v_{i+1}$ is the $j^{th}$ child of $v_i$ then the $j^{th}$ port of $\filterbed(v_i)$ is tricolored.
Now consider $\filterbed(v_5)$ which is vertical and note that its $k^{th}$ port is crossed by a tricolored port of $v_{k}$. Since one of the ports of $v_5$ must be tricolored, we obtain the pair $(H, V)$ with the desired properties. The proof follows.
\end{proof}

\section{Conclusion and open problems}\label{sec:conclusion}

We have given a simple construction that disproves Wegner's conjecture and shows that the standard LP relaxation with clique constraints for MISR has integrality gap at least $5/2-\varepsilon$ for every $\varepsilon>0$.  
Several questions remain open.
\begin{itemize}
\item What is the true supremum of $\tau(\Rcal)/\nu(\Rcal)$ for axis-parallel rectangle families $\Rcal$?  
\item What is the exact integrality gap of the LP for MISR? 
\item Is there a finite package with rectangles to port ratio exactly $5/2$ or larger? 
\end{itemize}

\section*{Acknowledgements}
Constructions of the initial counterexamples relied heavily on trial and error. OpenAI's GPT-5.5 Pro was used extensively to search for suitable constructions, often by generating code for these searches and for auxiliary verification in Lean. Once promising constructions emerged, the authors distilled the ideas into clear, concise, human-readable proofs. The authors also used integer programs generated by GPT-5.5 Pro to verify optimality of the constructions; however, the correctness of the proofs does not rely on these auxiliary verifications. OpenAI Codex was also used for drawing the figures and for initial drafts of several portions of the text. Anthropic's Claude Pro was used to identify typographical errors. All mathematical proofs and the final text were reviewed and rewritten by the authors.

\bibliographystyle{plain}
\bibliography{refs}

@inproceedings{Chalermsook11,
  author       = {Parinya Chalermsook},
  editor       = {Leslie Ann Goldberg and
                  Klaus Jansen and
                  R. Ravi and
                  Jos{\'{e}} D. P. Rolim},
  title        = {Coloring and Maximum Independent Set of Rectangles},
  booktitle    = {Approximation, Randomization, and Combinatorial Optimization. Algorithms
                  and Techniques - 14th International Workshop, {APPROX} 2011, and 15th
                  International Workshop, {RANDOM} 2011, Princeton, NJ, USA, August
                  17-19, 2011. Proceedings},
  series       = {Lecture Notes in Computer Science},
  volume       = {6845},
  pages        = {123--134},
  publisher    = {Springer},
  year         = {2011},
  url          = {https://doi.org/10.1007/978-3-642-22935-0\_11},
  doi          = {10.1007/978-3-642-22935-0\_11},
  bibsource    = {dblp computer science bibliography, https://dblp.org}
}

@inproceedings{ChalermsookW21,
  author       = {Parinya Chalermsook and
                  Bartosz Walczak},
  editor       = {D{\'{a}}niel Marx},
  title        = {Coloring and Maximum Weight Independent Set of Rectangles},
  booktitle    = {Proceedings of the 2021 {ACM-SIAM} Symposium on Discrete Algorithms,
                  {SODA} 2021, Virtual Conference, January 10 - 13, 2021},
  pages        = {860--868},
  publisher    = {{SIAM}},
  year         = {2021},
  url          = {https://doi.org/10.1137/1.9781611976465.54},
  doi          = {10.1137/1.9781611976465.54},
  bibsource    = {dblp computer science bibliography, https://dblp.org}
}

@article{wegner1965,
  author  = {Wegner, Gerd},
  title   = {\"{U}ber eine kombinatorisch-geometrische Frage von Hadwiger und Debrunner},
  journal = {Israel Journal of Mathematics},
  volume  = {3},
  pages   = {187--198},
  year    = {1965}
}

@article{gyarfasLehel1985,
  author  = {Gy{\'a}rf{\'a}s, Andr{\'a}s and Lehel, Jen{\H{o}}},
  title   = {Covering and Coloring Problems for Relatives of Intervals},
  journal = {Discrete Mathematics},
  volume  = {55},
  number  = {2},
  pages   = {167--180},
  year    = {1985}
}

@article{karolyi1991,
  author  = {K{\'a}rolyi, Gyula},
  title   = {On Point Covers of Parallel Rectangles},
  journal = {Periodica Mathematica Hungarica},
  volume  = {23},
  number  = {2},
  pages   = {105--107},
  year    = {1991}
}

@article{karolyiTardos1996,
  author  = {K{\'a}rolyi, Gyula and Tardos, G{\'a}bor},
  title   = {On Point Covers of Multiple Intervals and Axis-Parallel Rectangles},
  journal = {Combinatorica},
  volume  = {16},
  number  = {2},
  pages   = {213--222},
  year    = {1996}
}

@article{fonderflaassKostochka1993,
  author  = {Fon-Der-Flaass, Dmitry G. and Kostochka, Alexandr V.},
  title   = {Covering Boxes by Points},
  journal = {Discrete Mathematics},
  volume  = {120},
  number  = {1--3},
  pages   = {269--275},
  year    = {1993}
}

@article{aronovEzraSharir2010,
  author  = {Aronov, Boris and Ezra, Esther and Sharir, Micha},
  title   = {Small-Size {$\epsilon$}-Nets for Axis-Parallel Rectangles and Boxes},
  journal = {SIAM Journal on Computing},
  volume  = {39},
  number  = {7},
  pages   = {3248--3282},
  year    = {2010}
}

@article{correaEtAl2015,
  author  = {Correa, Jos{\'e} R. and Feuilloley, Laurent and P{\'e}rez-Lantero, Pablo and Soto, Jos{\'e} A.},
  title   = {Independent and Hitting Sets of Rectangles Intersecting a Diagonal Line: Algorithms and Complexity},
  journal = {Discrete \& Computational Geometry},
  volume  = {53},
  number  = {2},
  pages   = {344--365},
  year    = {2015}
}

@article{chudnovskySpirklZerbib2018,
  author  = {Chudnovsky, Maria and Spirkl, Sophie and Zerbib, Shira},
  title   = {Piercing Axis-Parallel Boxes},
  journal = {The Electronic Journal of Combinatorics},
  volume  = {25},
  number  = {1},
  pages   = {P1.70},
  year    = {2018}
}

@article{chenDumitrescu2020,
  author  = {Chen, Ke and Dumitrescu, Adrian},
  title   = {On Wegner's Inequality for Axis-Parallel Rectangles},
  journal = {Discrete Mathematics},
  volume  = {343},
  number  = {12},
  pages   = {112091},
  year    = {2020}
}

@article{tomon2023,
  author  = {Tomon, Istv{\'a}n},
  title   = {Lower Bounds for Piercing and Coloring Boxes},
  journal = {Advances in Mathematics},
  volume  = {435},
  pages   = {109360},
  year    = {2023}
}

@inproceedings{chalermsookChuzhoy2009,
  author    = {Chalermsook, Parinya and Chuzhoy, Julia},
  title     = {Maximum Independent Set of Rectangles},
  booktitle = {Proceedings of the Twentieth Annual ACM-SIAM Symposium on Discrete Algorithms (SODA)},
  pages     = {892--901},
  publisher = {SIAM},
  year      = {2009}
}

@inproceedings{adamaszekWiese2013,
  author    = {Adamaszek, Anna and Wiese, Andreas},
  title     = {Approximation Schemes for Maximum Weight Independent Set of Rectangles},
  booktitle = {Proceedings of the 54th IEEE Annual Symposium on Foundations of Computer Science (FOCS)},
  pages     = {400--409},
  publisher = {IEEE Computer Society},
  year      = {2013}
}

@inproceedings{mitchell2021,
  author    = {Mitchell, Joseph S. B.},
  title     = {Approximating Maximum Independent Set for Rectangles in the Plane},
  booktitle = {Proceedings of the 62nd IEEE Annual Symposium on Foundations of Computer Science (FOCS)},
  pages     = {339--350},
  publisher = {IEEE Computer Society},
  year      = {2021}
}

@inproceedings{galvezEtAl2022,
  author    = {G{\'a}lvez, Waldo and Khan, Arindam and Mari, Mathieu and M{\"o}mke, Tobias and Pittu, Madhusudhan Reddy and Wiese, Andreas},
  title     = {A 3-Approximation Algorithm for Maximum Independent Set of Rectangles},
  booktitle = {Proceedings of the 2022 ACM-SIAM Symposium on Discrete Algorithms (SODA)},
  pages     = {894--905},
  publisher = {SIAM},
  year      = {2022}
}

@article{caoduroEtAl2023,
  author  = {Caoduro, Marco and Cslovjecsek, Jana and Pilipczuk, Micha{\l} and W{\k{e}}grzycki, Karol},
  title   = {On the Independence Number of Intersection Graphs of Axis-Parallel Segments},
  journal = {Journal of Computational Geometry},
  volume  = {14},
  number  = {1},
  pages   = {144--156},
  year    = {2023}
}

@article{fowlerPatersonTanimoto1981,
  author  = {Fowler, Robert J. and Paterson, Michael S. and Tanimoto, Steven L.},
  title   = {Optimal Packing and Covering in the Plane are {NP}-Complete},
  journal = {Information Processing Letters},
  volume  = {12},
  number  = {3},
  pages   = {133--137},
  year    = {1981}
}

@inproceedings{chuzhoyEne2016,
  author    = {Chuzhoy, Julia and Ene, Alina},
  title     = {On Approximating Maximum Independent Set of Rectangles},
  booktitle = {Proceedings of the 57th IEEE Annual Symposium on Foundations of Computer Science (FOCS)},
  pages     = {820--829},
  publisher = {IEEE Computer Society},
  year      = {2016}
}

@article{AsplundGrunbaum1960,
  author  = {Asplund, Edgar and Gr{\"u}nbaum, Branko},
  title   = {On a coloring problem},
  journal = {Mathematica Scandinavica},
  volume  = {8},
  pages   = {181--188},
  year    = {1960},
  doi     = {10.7146/math.scand.a-10607}
}

@book{Matousek2002,
  author    = {Ji{\v{r}}{\'\i} Matou{\v{s}}ek},
  title     = {Lectures on Discrete Geometry},
  series    = {Graduate Texts in Mathematics},
  volume    = {212},
  publisher = {Springer},
  address   = {New York},
  year      = {2002},
  isbn      = {978-0-387-95373-1}
}

@book{Berge1989,
  author    = {Claude Berge},
  title     = {Hypergraphs: Combinatorics of Finite Sets},
  series    = {North-Holland Mathematical Library},
  volume    = {45},
  publisher = {North-Holland},
  address   = {Amsterdam},
  year      = {1989},
  isbn      = {9780444874894}
}

@article{PachTardos2013,
  author  = {J{\'a}nos Pach and G{\'a}bor Tardos},
  title   = {Tight Lower Bounds for the Size of Epsilon-Nets},
  journal = {Journal of the American Mathematical Society},
  volume  = {26},
  number  = {3},
  pages   = {645--658},
  year    = {2013},
  doi     = {10.1090/S0894-0347-2012-00759-0}
}

@article{GalvezKhanMariMomkeReddyWiese2021,
  author        = {Waldo G{\'a}lvez and
                   Arindam Khan and
                   Mathieu Mari and
                   Tobias M{\"o}mke and
                   Madhusudhan Reddy and
                   Andreas Wiese},
  title         = {A {(2+$\epsilon$)}-Approximation Algorithm for Maximum Independent Set of Rectangles},
  journal       = {CoRR},
  volume        = {abs/2106.00623},
  year          = {2021},
  archivePrefix = {arXiv},
  eprint        = {2106.00623},
  primaryClass  = {cs.CG},
  url           = {https://arxiv.org/abs/2106.00623}
}

@article{gyarfas1987problems,
  title={Problems from the world surrounding perfect graphs},
  author={Gy{\'a}rf{\'a}s, Andr{\'a}s},
  journal={Applicationes Mathematicae},
  volume={19},
  number={3-4},
  pages={413--441},
  year={1987},
  publisher={Polska Akademia Nauk. Instytut Matematyczny PAN}
}

@article{PAWLIK20146,
title = {Triangle-free intersection graphs of line segments with large chromatic number},
journal = {Journal of Combinatorial Theory, Series B},
volume = {105},
pages = {6-10},
year = {2014},
issn = {0095-8956},
doi = {https://doi.org/10.1016/j.jctb.2013.11.001},
url = {https://www.sciencedirect.com/science/article/pii/S009589561300083X},
author = {Arkadiusz Pawlik and Jakub Kozik and Tomasz Krawczyk and Michał Lasoń and Piotr Micek and William T. Trotter and Bartosz Walczak}
}

@article{scott2020survey,
  title={A survey of $\chi$-boundedness},
  author={Scott, Alex and Seymour, Paul},
  journal={Journal of Graph Theory},
  volume={95},
  number={3},
  pages={473--504},
  year={2020},
  publisher={Wiley Online Library}
}

\end{document}